\documentclass[pra,twocolumn,superscriptaddress]{revtex4-1}
\usepackage{graphicx}
\usepackage{dcolumn}
\usepackage{bm}
\usepackage{amssymb}
\usepackage{braket}
\usepackage{lipsum}
\usepackage{amsmath}
\usepackage{color}
\usepackage[colorlinks=true,linkcolor=blue,citecolor=red,plainpages=false,pdfpagelabels]{hyperref}
\usepackage{cleveref}

\newcommand{\dya}[1]{\ket{#1}\!\bra{#1}}
\newcommand{\Tr}{{\rm Tr}}

\newtheorem{theorem}{Theorem}
\newtheorem{proposition}[theorem]{Proposition}
\newenvironment{proof}[1][Proof]{\noindent\textbf{#1.} }{\ \rule{0.5em}{0.5em}}

\begin{document}

\title{Entropic Energy-Time Uncertainty Relation}
\author{Patrick J. Coles}
\affiliation{Theoretical Division, Los Alamos National Laboratory, Los Alamos, New Mexico 87545, USA}
\author{Vishal Katariya}
\affiliation{Hearne Institute for Theoretical Physics, Department of Physics and Astronomy, Louisiana State University, Baton Rouge, Louisiana 70803, USA}
\author{Seth Lloyd}
\affiliation{Department of Mechanical Engineering, Massachusetts Institute of Technology, Cambridge, Massachusetts 02139, USA}
\affiliation{Research Laboratory of Electronics, Massachusetts Institute of Technology, Cambridge, Massachusetts 02139, USA}
\author{Iman Marvian}
\affiliation{Departments of Physics \& Electrical and Computer Engineering, Duke University, Durham, North Carolina 27708, USA}
\author{Mark M. Wilde}
\affiliation{Hearne Institute for Theoretical Physics, Department of Physics and Astronomy, Louisiana State University, Baton Rouge, Louisiana 70803, USA}
\affiliation{Center for Computation and Technology, Louisiana State University, Baton Rouge, Louisiana 70803, USA}
\date{\today}

\begin{abstract}
Energy-time uncertainty plays an important role in quantum foundations and technologies, and it was even discussed by the founders of quantum mechanics. However, standard approaches (e.g., Robertson's uncertainty relation) do not apply to energy-time uncertainty because, in general, there is no Hermitian operator associated with time. Following previous approaches, we quantify time uncertainty by how well one can read off the time from a quantum clock. We then use entropy to quantify the information-theoretic distinguishability of the various time states of the clock. Our main result is an entropic energy-time uncertainty relation for general time-independent Hamiltonians, stated for both the discrete-time and continuous-time cases. Our uncertainty relation is strong, in the sense that it allows for a quantum memory to help reduce the uncertainty, and this formulation leads us to reinterpret it as a bound on the relative entropy of asymmetry. Due to the operational relevance of entropy, we anticipate that our uncertainty relation will have information-processing applications.
\end{abstract}

\maketitle

\textit{Introduction}---The uncertainty principle is one of the most iconic implications of quantum mechanics, stating that there are pairs of observables that cannot be simultaneously known. It was first proposed by Heisenberg \cite{heisenberg} for the position $\hat{q}$ and momentum $\hat{p}$ observables and then rigorously stated by Kennard \cite{Kennard1927} in the familiar form using standard deviations:
$
\label{eqnKennard}
\Delta \hat{q} \Delta \hat{p} \geq \hbar/2$.
Robertson \cite{robertson} later formulated a similar relation for a different class of observables, namely, for pairs of bounded Hermitian observables $\hat{X}$ and $\hat{Z}$ (e.g., the Pauli spin operators), as
$
\label{eqnRobertson}
\Delta \hat{X} \Delta \hat{Z} \geq \frac{1}{2} |\langle [\hat{X},\hat{Z}] \rangle |$.
Since then, many alternative formulations have been proven for similar Hermitian operator pairs (e.g., \cite{Busch2014, Maccone2014}).

Unfortunately, these relations do not apply to energy and time since time does not, in general, correspond to a Hermitian operator. In particular, Pauli's theorem states that the semi-boundedness of a Hamiltonian precludes the existence of a Hermitian time operator, or in other words, if there was such an operator, then the Hamiltonian would be unbounded from below and thus unphysical \cite{pauli}.  Hence, formulating a general energy-time uncertainty relation is a nontrivial task. We point to \cite{butterfield} for an overview on time in quantum mechanics.

Nevertheless, the energy-time pair is of significant importance both fundamentally and technologically. Energy-time uncertainty was already discussed by the founders of quantum mechanics: Bohr, Heisenberg, Schr\"odinger, and Pauli (see \cite{Dodonov2015} for a review).  In the special case of the harmonic oscillator, this pair corresponds to number and phase, and number-phase uncertainty is relevant to metrology \cite{Giovannetti2004}, e.g., phase estimation in interferometry. The energy-time pair is arguably the most general observable pair in the sense that it applies to all physical systems (i.e., all systems have a Hamiltonian). 
 
Despite the lack of a Hermitian observable associated with time, relations with the feel of energy-time uncertainty relations have been formulated. Mandelstam and Tamm \cite{mandelstamm-tamm} related the energy standard deviation $\Delta E$ to the time $\tau$ that it takes for a state to move to an orthogonal state:
$
\label{eqnMT}
\tau \Delta E\geq \frac{\pi \hbar}{2}\,.
$
This relation can be thought of as a speed limit---a bound on how fast a quantum state can move---and other similar speed limits have been formulated \cite{margolus-levitin}. Alternatively, it can be thought of as bounding how well a quantum system acts as a clock, since the time resolution of the clock is related to the time $\tau$ for the system to move to an orthogonal state.

In this work, we take the clock perspective on time uncertainty: one's uncertainty about time corresponds to how well one can ``read off'' the time from measuring a quantum clock. 
A natural measure for this purpose is to consider the information-theoretic distinguishability of the various time states. As such, we propose using entropy to quantify time uncertainty, and our main result is an entropic energy-time uncertainty relation.

Entropy has been widely employed in uncertainty relations for position-momentum \cite{hirschman} and finite-dimensional observables \cite{deutsch, maassen-uffink}---see \cite{entropic-ur-review} for a recent detailed review of entropic uncertainty relations. The key benefits of entropy as an uncertainty measure are its clear operational meaning and its relevance to information-processing applications. Indeed, entropic uncertainty relations form the cornerstone of security proofs for quantum key distribution and other quantum cryptographic tasks \cite{entropic-ur-review}. They furthermore allow one to recast the uncertainty principle in terms of a guessing game, as we do below for  energy and time.

An entropic uncertainty relation for energy and time was previously given in \cite{hall} by constructing an almost-periodic time observable and using a so-called almost-periodic entropy for time. This approach was extended in \cite{Hall2018}, where the Holevo information bound was used to derive an entropic energy-time uncertainty relation. However, as indicated in \cite{hall}, an almost-periodic time observable serves as a poor quantum clock for aperiodic systems. In \cite{Boette2016}, the entanglement between a system and a clock was used to derive an entropic energy-time uncertainty relation for a Hamiltonian with a uniformly spaced spectrum.

In this paper, we derive entropic energy-time uncertainty relations for general, time-independent Hamiltonians. We first derive a relation for discrete and arbitrarily spaced time, and then we extend this relation to infinitesimally closely spaced (i.e., continuous) time. Our results apply to systems with either finite- or infinite-dimensional Hamiltonians.

A novel aspect of our energy-time uncertainty relation is that it allows the observer to reduce their uncertainty through access to a quantum memory system, as was the case in prior uncertainty relations \cite{berta}. The two main benefits of allowing for quantum memory are that (1) it dramatically tightens the relation when the clock is in a mixed state, and (2) it makes the relation more relevant to cryptographic applications in which the eavesdropper may hold the memory system (e.g., see \cite{berta}). Furthermore, by allowing for quantum memory, we can reinterpret our uncertainty relation as a bound on the relative entropy of asymmetry \cite{GMS09}, and we discuss below the implications of this reinterpretation.

The fact that our uncertainty relation is stated using operationally-relevant entropies implies that it should be useful for information processing applications. For example, if one can distinguish between the time states well, then it is possible to extract randomness by performing an energy measurement. True random bits are critical to the execution of secure protocols and numerical computations. In this case, the randomness of energy measurement outcomes is certified by our bound. Entropic uncertainty relations also find use in proving the security of quantum key distribution (QKD) protocols \cite{BB84}. If one party is able to prepare states in both the phase and number bases of photons, and if another party is able to perform measurements in these two bases, then both parties can distill a secret key whose security is guaranteed by our relation. We provide more details regarding applications in the supplementary material (Appendix \ref{appF}).

\begin{figure}
\begin{center}
\includegraphics[width=2.5in]{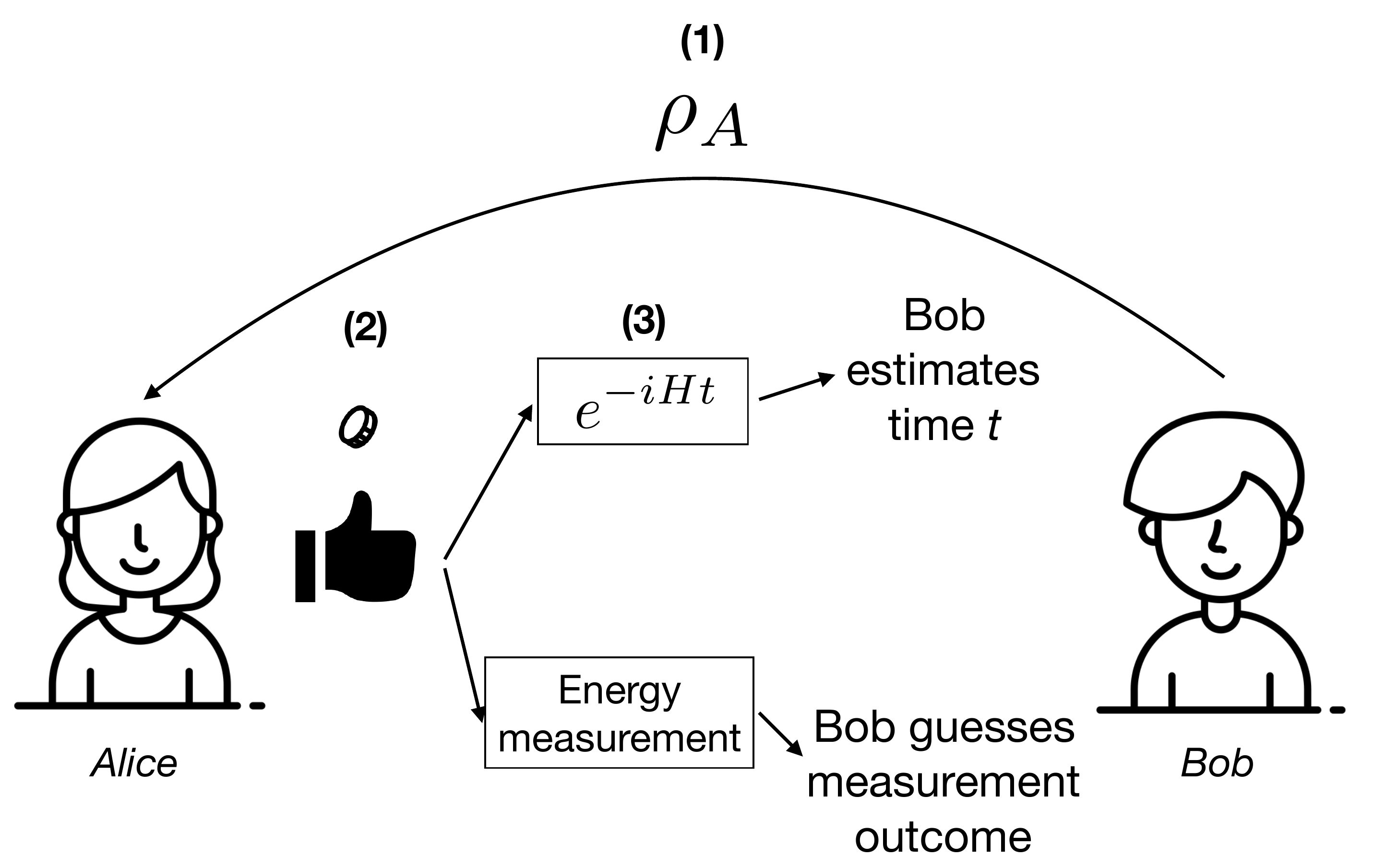}
    \caption{Guessing game for energy-time uncertainty. (1) Bob prepares a quantum clock
    %system $A$
    in the state $\rho_A$ and sends it to Alice. (2) Alice flips a coin and (3) either measures the clock's energy or randomly sets the clock's time (i.e., applies a time evolution $e^{-iHt}$ with $t$ randomly  chosen from a predefined set). Bob's goal is to, depending on Alice's coin flip, guess the clock's energy or guess $t$ by reading the clock. Our uncertainty relations constrain Bob's ability to win this game.}
\label{fig:guessing-game}
\end{center}
\end{figure}

Uncertainty relations can be understood in the framework of a guessing game involving two players, Alice and Bob \cite{berta, entropic-ur-review}, and Figure~\ref{fig:guessing-game} shows this game for the energy-time pair. Bob prepares system $A$ in an arbitrary state $\rho_A$ and sends it to Alice. Alice then flips a coin. If she gets heads, she performs an energy measurement, and Bob then must guess the outcome (possibly with the help of a memory system $R$ that is initially correlated to $A$). If she gets tails, she applies a time evolution $e^{-iHt}$ in which $t$ is randomly chosen from some predefined set, and then sends $A$ back to Bob, who then tries to guess which time $t$ Alice applied.  All of our uncertainty relations can be understood in terms of this guessing game and can be viewed as constraints on Bob's probability of winning this game (i.e., guessing both the energy and time correctly). There are other variations of this energy-time uncertainty guessing game that are possible, one of which is discussed in the Supplementary Material (Appendix~\ref{appA}).

In what follows, we give some necessary preliminaries before stating our main result for the R\'enyi entropy family in the discrete-time case, and then we extend to the continuous-time case for the von Neumann entropy. Finally, we apply our relation to an illustrative example of a spin-1/2 particle.

\textit{Preliminaries}---We begin by considering a finite-dimensional Hamiltonian $H$ that acts on a quantum system $A$, and suppose that it has $N_E \in \mathbb{Z^+}$ real energy eigenvalues taken from a set $\mathcal{E}\subset \mathbb{R}$. We thus write the Hamiltonian as
$
H_A=\sum_{\varepsilon\in\mathcal{E}}\varepsilon\Pi_A^{\varepsilon},\label{eq:ham}
$
where $\Pi_A^{\varepsilon}$ denotes the projector onto the subspace spanned by energy eigenstates with eigenvalue $\varepsilon$. The projectors obey $\Pi^{\varepsilon}\Pi^{\varepsilon^{\prime}}
=\Pi^{\varepsilon}\delta_{\varepsilon,\varepsilon^{\prime}}$, where $\delta_{\varepsilon,\varepsilon^{\prime}} = 1$ if $\varepsilon = \varepsilon^{\prime}$ and $\delta_{\varepsilon,\varepsilon^{\prime}} = 0$ otherwise.

We now recall how to encode the classical state of a clock into a quantum system. Inspired by the Feynman-Kitaev history state formalism \cite{feynman2, F86, kitaev}, as well as the quantum time proposal of \cite{GLM15}, we introduce a register $T$ for storing the time, which can be interpreted as a background reference clock. A measurement on the time register is treated in this framework as a time measurement.
 Let $\mathcal{T} = \{t_1, \ldots, t_K\}$ denote a set of times, for integer $K\geq 2$, such that $t_k \in \mathbb{R}$ for all $k \in \{1, \ldots, K\}$, and $t_1 \leq t_2 \leq \ldots \leq t_K$.
 We suppose that the register $T$ has a complete, discrete, and orthonormal basis $\left\{  |t_k\rangle\right\}  _{k=1}^{K}$. The time values need not be evenly spaced, which means that the basis for register $T$ can include any combination of $|\mathcal{T}| = K$ distinct and orthonormal kets.

Now consider a clock system $A$ that may initially be correlated to a memory system $R$, together in a joint state $\rho_{AR}$ with $\rho_A = \Tr_R{(\rho_{AR})}$. Let random variable $E$ capture the outcomes of an energy measurement on the system $A$. The outcomes can be stored in a classical register, which we also denote without ambiguity by $E$ in what follows. To quantify energy uncertainty, we employ the R\'enyi conditional entropy  $S_{\alpha}(E|R)$ (defined below) of the following classical-quantum state:
\begin{align}
\omega_{ER}
\equiv\sum_{\varepsilon \in \mathcal{E}}
|\varepsilon\rangle\langle \varepsilon|_{E}\otimes\operatorname{Tr}%
_{A}\{\Pi^{\varepsilon}_A\rho_{AR}\},
\label{eq:omega-state}
\end{align}
where the kets $\{\ket{\varepsilon}\}_{\varepsilon \in \mathcal{E}}$ are orthonormal, obeying $\langle \varepsilon' | \varepsilon \rangle = \delta_{\varepsilon' , \varepsilon}$, and thus serve as classical labels for the energies of the Hamiltonian. To quantify the time uncertainty, we employ the R\'enyi conditional entropy $S_{\alpha}(T|A)$ of the following classical-quantum state:
\begin{align}
\kappa_{TA}
&
\equiv\frac{1}{|\mathcal{T}|}\sum_{k=1}^{K}|t_k\rangle\langle t_k|_{T}\otimes e^{-iHt_k}\rho_{A}e^{iHt_k}.
\label{eq:kappa-state}
\end{align}
In the above and henceforth, we set $\hbar = 1$.  The state $\kappa_{TA}$ can be interpreted as the joint state of system $A$ (the local quantum clock) and the background reference clock $T$, at an unknown time $t_k\in \mathcal{T}$ chosen according to the uniform distribution.  Equivalently, this state  can be understood as a time-decohered version of the Feynman-Kitaev history state \cite{feynman2, F86, kitaev}, the latter of which has the entire history of the state $\rho_A(t)$ encoded and entangled with a time register in superposition. The classical-quantum states in \eqref{eq:omega-state} and \eqref{eq:kappa-state} are in one-to-one correspondence with the following labeled ensembles, respectively:
\begin{align*}
\{ p(\varepsilon), \ |\varepsilon\rangle\langle \varepsilon|_{E} \otimes \operatorname{Tr}%
_{A}\{\Pi^{\varepsilon}_A\rho_{AR}\} / p(\varepsilon)\}_{\varepsilon \in \mathcal{E}},\\
\{ 1/ |\mathcal{T}|, \ |t_k\rangle\langle t_k|_{T} \otimes e^{-iHt_k}\rho_{A}e^{iHt_k}\}_{t_k \in \mathcal{T}}, 
\end{align*}
where $p(\varepsilon) = \operatorname{Tr}\{\Pi^{\varepsilon}_A\rho_{AR}\}$.

\textit{R\'enyi entropies}---For a probability distribution $\{p_j\}$, the R\'enyi entropies are defined for $\alpha \in (0,1)\cup(1,\infty)$ by
$
S_{\alpha}(\{p_j\} ) = \frac{1}{1-\alpha} \log_2 \sum_j p_j^{\alpha}\,,
$
and for $\alpha\in\{0,1,\infty\}$ in the limit.
This entropy family is generalized to quantum states via the sandwiched R\'enyi conditional entropy \cite{Mueller2013}, defined for a bipartite state $\rho_{AB}$ with $\alpha \in (0,\infty]$ as
\begin{equation}
S_{\alpha}(A|B)_{\rho} = - \inf_{\sigma_B} D_{\alpha}(\rho_{AB} \Vert I_A \otimes \sigma_B) ,
\end{equation}
where the optimization is with respect to all density operators $\sigma_B$ on system $B$. The quantity $S_{\alpha}(A|B)_{\rho}$ is in turn defined from the sandwiched R\'enyi relative entropy of a density operator $\xi$ and a positive semi-definite operator~$\zeta$, which is defined for $\alpha\in(0,1)\cup(1,\infty)$ as \cite{Mueller2013,WWY13}
\begin{equation}
D_{\alpha}(\xi \Vert \zeta) = 
      \frac{1}{\alpha-1} \log_2  \text{Tr}\big[ (\zeta^{\frac{1-\alpha}{2\alpha}} \xi \zeta^{\frac{1 - \alpha}{2 \alpha}} )^\alpha \big] .
\end{equation}
If $\alpha>1$ and the support of $\xi$ is not contained in the support of $\zeta$, then it is defined to be equal to $+\infty$. The sandwiched R\'enyi relative entropy $D_{\alpha}(\xi \Vert \zeta)$ is defined for $\alpha \in \{1,\infty\}$ in the limit.

\textit{Entropic energy-time uncertainty relation}---Let us now state our uncertainty relation for energy and time. For a pure state $\rho_A=\dya{\psi}_A$ uncorrelated with a reference system $R$, it is as follows:
\begin{gather}
    S_{\alpha}(T|A)_{\kappa} + S_{\beta}( \{p(\varepsilon) \} ) \geq\log_2 |\mathcal{T}|,
\label{eq:pure-state-inequality}
\end{gather}
holding for all $\alpha\in [1/2,\infty]$, with $\beta$ satisfying $1/\alpha + 1/\beta =2$, 
where $p(\varepsilon) = \langle \psi | \Pi^{\varepsilon}_A | \psi \rangle$. The above inequality \eqref{eq:pure-state-inequality} is saturated, e.g., when $\ket{\psi}$ is an energy eigenstate. Such states also maximize the time uncertainty, $S_{\alpha}(T|A)_{\kappa} = \log_2 |\mathcal{T}|$, since they are stationary states. 

The concavity of entropy and concavity of conditional entropy \cite{PhysRevLett.30.434} then directly imply that the same inequality in \eqref{eq:pure-state-inequality} holds for a mixed state uncorrelated with a reference system $R$. However, if $\rho_A$ is a maximally mixed state, the inequality in \eqref{eq:pure-state-inequality} yields a trivial bound on the total uncertainty. This is because the inequality does not capture the inherent uncertainty of the initial state.

%., as expected, since such states are stationary states, and by definition, their time evolution is trivial.

One of our main results remedies this deficiency, capturing the inherent uncertainty mentioned above and holding nontrivially for mixed states:
\begin{equation}
   S_{\alpha}(T|A)_{\kappa} + S_{\beta}(E|R)_{\omega} \geq\log_2 |\mathcal{T}|. \label{eq:energy-time-uncertainty}
\end{equation}
The entropic energy-time uncertainty relation in \eqref{eq:energy-time-uncertainty} holds for all $\alpha \in [1/2,\infty]$, where $\beta$ satisfies
$1/\alpha + 1/\beta = 2$, with the proof given in Appendix \ref{appB}. The quantity $S_{\alpha}(T|A)_{\kappa}$ represents the uncertainty about the time $t_k$ from the perspective of someone holding the $A$ system of the state $\kappa_{TA}$ in \eqref{eq:kappa-state}. The quantity $S_{\beta}(E|R)_{\omega}$, which is determined by the state $\rho_{AR}$ and the Hamiltonian $H_A$,  represents the uncertainty about the outcome of an energy measurement from the perspective of someone who possesses the $R$ system of the state $\omega_{ER}$ in \eqref{eq:omega-state}. In the case that $\rho_{AR}$ is pure, then the quantity $S_{\beta}(E|R)_{\omega}$ is determined by the reduced state $\rho_{A}$ and the Hamiltonian $H_A$. According to \eqref{eq:energy-time-uncertainty}, a good quantum clock state $\rho_A$, for which $S_{\alpha}(T|A)_{\kappa} \approx 0$, necessarily has a large uncertainty in the energy measurement, in the sense that $S_{\beta}(E|R)_{\omega} \gtrsim \log_2 |\mathcal{T}|$. Conversely, a state with a small uncertainty in the energy measurement, in the sense that $S_{\beta}(E|R)_{\omega}
 \approx 0$, is necessarily a poor quantum clock state, in the sense that $S_{\alpha}(T|A)_{\kappa} \approx \log_2 |\mathcal{T}|$.
 
Note that the uncertainties in \eqref{eq:energy-time-uncertainty} are entropic and hence do not quantify the uncertainties of time and energy in their units, but rather the amount of information (in bits) that we do not know about the respective quantities. For example, if a system can equally likely take on one of two energies $E_1$ and $E_2$, then the entropic uncertainty in energy constitutes only one bit, and it does not depend on the magnitudes of $E_1$ or $E_2$. Each entropy in \eqref{eq:energy-time-uncertainty} is analogous to a guessing probability, which quantifies how well one can guess the time $t$ given the state $ \rho_A(t)$, or the energy given the ability to measure a memory system $R$. In fact, $S_{\alpha}(A|B)$ converges to the negative logarithm of the guessing probability as $\alpha \rightarrow \infty$ \cite{Konig2009, Mueller2013}.

Considering the special case of  $|\mathcal{T}|=2 $, one finds a simple, yet interesting corollary of \eqref{eq:energy-time-uncertainty}:  under the Hamiltonian $H_A$, a quantum state $\rho_A$ can evolve to a perfectly distinguishable state, only if $S_{\beta}(E|R)_{\omega}\ge 1 $ for  $\beta \in [1/2,\infty]$. In other words, for  $S_{\beta}(E|R)_{\omega}< 1 $, the orthogonalization time $\tau$  in the Mandelstam-Tamm bound %in \eqref{eqnMT}
is infinite, which cannot be seen using
%\eqref{eqnMT}
Mandelstam-Tamm or other standard quantum speed limits.

By means of a quantum memory, one can also reduce the time uncertainty instead of  only reducing the energy uncertainty. This can be accomplished by considering the memory system $R$ to be a bipartite system $R_1 R_2$. One can then write the uncertainty relation in \eqref{eq:energy-time-uncertainty} as follows:
\begin{equation}
S_{\alpha}(T|AR_1)_{\kappa} + S_{\beta}(E|R_2)_{\omega} \geq\log_2 |\mathcal{T}|, \label{eq:memory-reduced-inequality}
\end{equation}
with full details given in Appendix \ref{appTimeUncertainty}. This shows that the tightening of \eqref{eq:pure-state-inequality} to give \eqref{eq:energy-time-uncertainty} using quantum memory can reduce the uncertainties in both energy and time. We note that this rewriting is achieved only by relabeling systems, and is thus a consequence of our earlier result in \eqref{eq:energy-time-uncertainty}.

An important special case of \eqref{eq:energy-time-uncertainty} is $\alpha =\beta =  1$ where both entropies are the von Neumann conditional entropy. This results in the following entropic uncertainty relation:
\begin{gather}
S(T|A)_{\kappa} + S(E|R)_{\omega} \geq \log_2 |\mathcal{T}|,
\label{eq:von-neumann-inequality}
\end{gather}
where the von Neumann conditional entropy of a bipartite state $\tau_{CD}$ can be written as
$
S(C|D)_{\tau} = -\operatorname{Tr}[\tau_{CD} \log_2 \tau_{CD}] + \operatorname{Tr}[\tau_{D} \log_2 \tau_{D}]$.
In fact, we show in Appendix~\ref{appD} of the Supplementary Material that the following equality holds for the von Neumann case when $\rho_{AR}$ is pure:
\begin{equation*}
S(T|A)_{\kappa} + S(E|R)_{\omega}  =  \log_2 |\mathcal{T}|  + D(\kappa_{A} \Vert \sum_{\varepsilon} \Pi^{\varepsilon}\rho_A \Pi^{\varepsilon} ) .
\end{equation*}
As discussed in the Supplementary Material (Appendix~\ref{appD}), when $\rho_{AR}$ is pure, equality in \eqref{eq:von-neumann-inequality} is achieved  [equivalently, $D(\kappa_{A} \Vert \sum_{\varepsilon} \Pi^{\varepsilon}\rho_A \Pi^{\varepsilon} )=0$] if and only if 
\begin{align}
\frac{1}{| \mathcal{T} |}\sum_{k = 1}^{K}e^{-iHt_k}\rho_A e^{iHt_k}  =  \sum_{\varepsilon} \Pi^{\varepsilon}\rho_A \Pi^{\varepsilon} \,.
\label{eq:vn_equalitycondition}
\end{align}
One way to satisfy \eqref{eq:vn_equalitycondition} is if $[\rho_A,H]=0$, and hence the relation is tight for states $\rho_A$ that are diagonal in the energy eigenbasis. Another way to satisfy \eqref{eq:vn_equalitycondition} is if 
$
  \frac{1}{|\mathcal{T}|} \sum_{k = 1}^{K}  e^{i (\varepsilon - \varepsilon')t_k}  = \delta_{\varepsilon, \varepsilon'}
$
for all combinations of $\varepsilon, \varepsilon'$. If the $ | \mathcal{T} | $ times are equally spaced, this implies that 
$
  e^{i (\varepsilon - \varepsilon') t_K} = 1$ and $ 
  (\varepsilon - \varepsilon') t_K = 2 \pi$.
This can be understood as an exact inverse relationship between the conjugate variables, which is a signature of a saturated uncertainty relation.

We remark that \eqref{eq:von-neumann-inequality} can be generalized to allow for non-uniform probabilities for the various times. As shown in the Supplementary Material (Appendix~\ref{appD}), the right-hand-side of \eqref{eq:von-neumann-inequality} gets replaced by the entropy $S(T)_{\kappa}$ of the time distribution  for this generalization.

\textit{Relative entropy of asymmetry formulation}---As shown in the Supplementary Material (Appendix~\ref{app:rel-ent-assym-to-cond-ent}), an alternative way of stating our main result in \eqref{eq:energy-time-uncertainty} is by employing the sandwiched R\'enyi relative entropy of asymmetry \cite{GJL17}, which generalizes an asymmetry measure put forward in \cite{GMS09}:
\begin{gather}
  S_{\alpha}(T|A)_{\kappa}+\inf_{\sigma:\left[  H,\sigma \right] =0} D_{\alpha}(\rho\Vert\sigma)\geq\log_2 |\mathcal{T}|, \label{eq:min-rel-ent-inequality}
\end{gather}
and holds for all $\alpha \in (0,\infty]$.
The inequality in \eqref{eq:min-rel-ent-inequality} delineates a trade-off, given the Hamiltonian $H$, between how well a state $\rho_A$ can serve as a quantum clock and the asymmetry of $\rho_A$ with respect to time translations. Moreover, this connection is exact for pure states. That is, a good quantum clock state $\rho_A$, for which $S_{\alpha}(T|A)_{\kappa} \approx 0$, is necessarily asymmetric with respect to time translations, in the sense that $\inf_{\sigma:\left[  H,\sigma\right] =0} D_{\alpha}(\rho\Vert\sigma) \gtrsim \log_2 |\mathcal{T}|$, which follows by employing \eqref{eq:min-rel-ent-inequality}. Conversely, a state that is nearly symmetric with respect to time translations, in the sense that $\inf_{\sigma:\left[  H,\sigma\right] =0} D_{\alpha}(\rho\Vert\sigma) \approx 0$, is necessarily a poor quantum clock state, in the sense that $S_{\alpha}(T|A)_{\kappa} \approx \log_2 |\mathcal{T}|$. This provides an alternate way of understanding the energy uncertainty of a pure state as asymmetry with respect to the generator of time evolutions. Note that $\log_2 |\mathcal{T}|-S_{\alpha}(T|A)_{\kappa}$ itself is a measure of asymmetry. Therefore,  \eqref{eq:min-rel-ent-inequality} can be interpreted as an upper bound on this measure of asymmetry in terms of R\'enyi relative entropy of asymmetry.

In the limit $\alpha \rightarrow 1$, the quantity $\inf_{\sigma:\left[  H,\sigma\right] =0}D_{\alpha}(\rho\Vert\sigma)$ reduces to
the relative entropy of asymmetry \cite{GMS09}
\begin{align}
\lim_{\alpha \to 1}
\inf_{\sigma:\left[  H,\sigma\right] =0} D_{\alpha}(\rho\Vert\sigma)
& = 
\inf_{\sigma:\left[  H,\sigma\right] =0} D(\rho\Vert\sigma) \notag \\
& \equiv 
\Gamma_H(\rho) = S(\Delta(\rho)) - S(\rho),
\end{align}
where the quantum relative entropy is defined as $D(\rho\Vert\sigma) \equiv \operatorname{Tr}[\rho [ \log_2 \rho - \log_2 \sigma]]$ \cite{U62} and $\Delta(\rho) = \sum_{\varepsilon \in \mathcal{E}} \Pi^{\varepsilon} \rho \Pi^{\varepsilon}$ (in the context of asymmetry, the function $S(\Delta(\rho)) - S(\rho)$ was first studied in \cite{vac2008}).
 Then the entropic uncertainty relation in \eqref{eq:min-rel-ent-inequality} reduces to
$
  S(T|A)_{\kappa} + \Gamma_H(\rho) \geq \log_2 |\mathcal{T}|. \label{eq:rea-inequality}
$

\begin{figure}
\begin{center}
\includegraphics[width=\columnwidth]{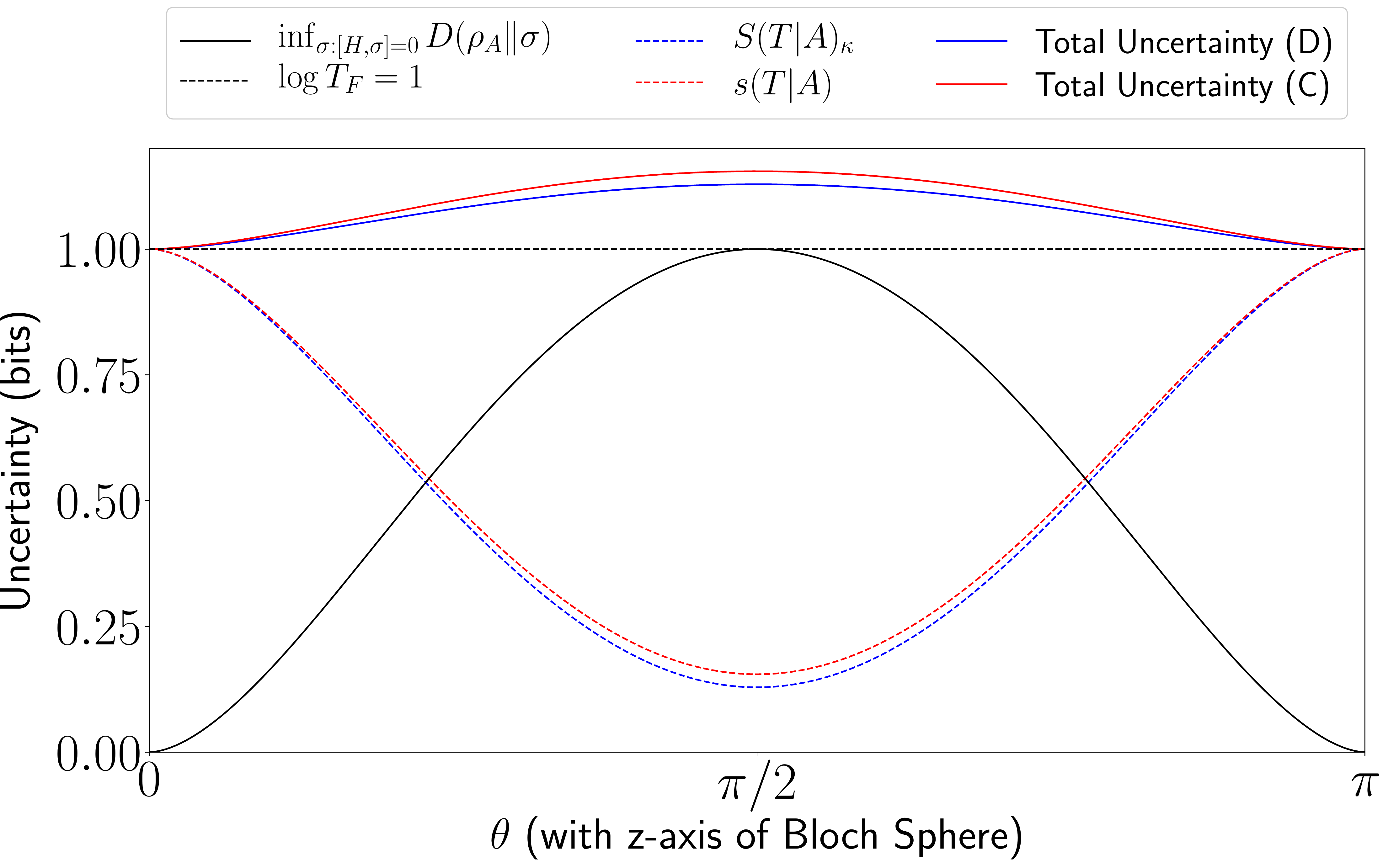}
\caption{Our uncertainty relations applied to a spin-1/2 particle in a magnetic field. For Hamiltonian $H = \kappa \sigma_z$ and  $| \mathcal{T} | = 2$ (in the discrete-time case) or $T_F = 2$ (in the continuous-time case), the plot shows the variation in the uncertainties with $\theta$, the angle the state makes with the $z$-axis of the Bloch sphere. The quantity $\inf_{\sigma:\left[  H,\sigma\right] =0} D(\rho_A\Vert\sigma)$ is the energy uncertainty, while $S(T|A)$ and $s(T|A)$ are respectively the time uncertainties for discrete and continuous time. The black dotted line shows $ \log_2 | \mathcal{T} | = \log_2 T_F  $, i.e., our lower bounds on the total uncertainty from \eqref{eq:energy-time-uncertainty} and \eqref{eq:continuous-time}.}
\label{fig:single-spin-uncertainty}
\end{center}
\end{figure}

\textit{Extension to continuous time}---We now extend the uncertainty relation in \eqref{eq:energy-time-uncertainty} so that it is applicable to continuous, as opposed to discrete, time, and to Hamiltonians with countable spectrum. From \eqref{eq:energy-time-uncertainty} and \cite{furrer}, we derive an inequality applicable to the von Neumann entropies. Full details are available in the Supplementary Material (Appendix~\ref{app:contin-time-extension}). Consider time to be continuous in the interval $\lbrack 0, T_F \rbrack$. Given a state $\rho_A$ and a Hamiltonian $H_A=\sum_{\varepsilon\in\mathcal{E}}\varepsilon\Pi_A^{\varepsilon}$, with $\mathcal{E}$ countably infinite, we then have that
\begin{gather}
  \inf_{\sigma:\left[  H,\sigma\right] =0} D(\rho_A\Vert\sigma) + s(T|A) \geq \log_2 T_F. \label{eq:continuous-time}
\end{gather}
For a continuously parametrized ensemble of states $\{p(x), \rho^x_B\}_{x\in \mathcal{X}}$, the differential conditional quantum entropy $s(X|B)$ is defined as 
$
s(X|B) = - \int_{\mathcal{X}} dx \ D(p(x) \rho^x_B \Vert \rho_{\text{avg}})$,
where $\rho_{\text{avg}} = \int_{\mathcal{X}} dx \ p(x) \rho^x_B$ \cite{furrer}. For our case, this means
\begin{align*}
s(T|A) & =-\int_{0}^{T_{F}}dt\ D(\rho(t)/T_{F}\Vert{\overline{\rho}}), \\
\overline{\rho} & =\frac{1}{T_{F}}\int_{0}^{T_{F}}dt\ e^{-iHt}\rho_A e^{iHt}.
\end{align*} 
We note here that there is an alternative way of phrasing the inequality in \eqref{eq:continuous-time} in dimensionless units  \footnote{The continuous time inequality can be rewritten as 
\unexpanded{$\inf_{\sigma:\left[  H,\sigma\right] =0} D(\rho_A\Vert\sigma)  \geq \log_2 (T_F/ 2^{s(T|A)})}$}.

\textit{Example: Spin in a magnetic field}---Consider a spin-$1/2$ particle in a magnetic field $ \mathbf{B} = B \hat{\mathbf{z}}$. This is described by the Hamiltonian $\hat{H} = \kappa \sigma_z$, where $\kappa$ is a constant proportional to $B$, and $\sigma_z$ is the $z$-Pauli operator. Consider a pure state $\rho_A = \dya{\psi(0)}$ that makes an angle $\theta$ with the $z$-axis of the Bloch sphere, given by $\ket{\psi(0)} = \cos(\theta/2) \ket{0} + \sin(\theta/2) \ket{1}$. After a time $t$, this state evolves to $\ket{\psi(t)} = e^{-i H t} \ket{\psi(0)}$. Figure~\ref{fig:single-spin-uncertainty} plots the variation of the uncertainty (time, energy, and total uncertainty) with $\theta$ for both our discrete- and continuous-time relations. For $\theta = \pi/2$, the energy uncertainty is maximal (one bit) while the time uncertainty is minimal (although still non-zero in this example). At the other extreme, for $\theta = 0$ or $\pi$, the energy uncertainty is zero while the time uncertainty is maximal (one bit), meaning that clock's time states cannot be distinguished. One can see in Figure~\ref{fig:single-spin-uncertainty} that our uncertainty relation is tight in this extreme case.

\textit{Discussion}---In this paper, we gave a conceptually clear and operational formulation of the energy-time uncertainty principle. We stated an entropic energy-time uncertainty relation for the R\'enyi entropies for discrete time sets. This relation was strengthened for mixed states by allowing the observer to possess a quantum memory, a feature that also allowed us to reinterpret our relation as a bound on the relative entropy of asymmetry. For the special case of von Neumann entropy, we extended our uncertainty relation to continuous time sets. Our relation is saturated for all states $\rho_A$ that are diagonal in the energy eigenbasis.

Expressed in terms of entropies, which are operationally important in information theory, our result should have uses in various tasks. Entropic uncertainty relations have been used previously to certify randomness and prove security of quantum cryptography protocols, and we believe our result will be an important tool used to develop such protocols further.

%The fact that our uncertainty relation is stated in terms of operationally-relevant entropies implies that it should be useful for information-processing applications. For example, we speculate that our uncertainty relation will be useful for proving the security of quantum cryptographic tasks, such as randomness extraction and quantum key distribution. We hope that our relation inspires the formulation of cryptographic protocols involving energy and time.

\begin{acknowledgments}
PJC acknowledges support from the Los Alamos National Laboratory ASC Beyond Moore's Law project.  VK acknowledges support from the Department of Physics and Astronomy at Louisiana State University. SL was supported by IARPA under the QEO program, NSF, and ARO under the Blue Sky Initiative. MMW acknowledges support from the US National Science Foundation through grant no.~1714215. MMW is grateful to SL for hosting him for a research visit to University of Oxford during May 2018, during which some of this research was conducted.  
\end{acknowledgments}
\bibliographystyle{unsrt}
%\bibliography{refs}

\pagebreak
\onecolumngrid

\appendix
\vspace{.4in}
\begin{center}
\textbf{SUPPLEMENTARY MATERIAL}
\end{center}
\section{Applications} \label{appF}

\subsection{Quantum speed limits}

In this section, we discuss the application of our uncertainty relation to quantum speed limits. Recall that the Mandelstam-Tamm speed limit \cite{mandelstamm-tamm}
%in \eqref{eqnMT}
has the form $\tau \Delta E\geq \frac{\pi \hbar}{2}$.

Our uncertainty relation gives a strong constraint on the time $\tau$ that appears in quantum speed limits, as follows. Specializing our discrete-time uncertainty relation to the case of $|\mathcal{T}| = 2$ gives
\begin{equation}
   S_{\alpha}(T|A)_{\kappa} \geq 1 -  S_{\beta}(E|R)_{\omega} . 
  \label{eq:ourspeedlimit}
\end{equation}
Here we can write $\mathcal{T} = \{0,t\}$. In this scenario, $S_{\alpha}(T|A)_{\kappa} =0 $ if and only if $\rho_A$ is orthogonal to $e^{-iHt}\rho_Ae^{iHt}$. Hence, the uncertainty relation in \eqref{eq:ourspeedlimit} implies the following:
\begin{equation}
   \tau \to \infty \quad \text{if}\quad S_{\beta}(E|R)_{\omega} < (1 \text{ bit})\,. 
  \label{eq:ourspeedlimit2}
\end{equation}
Here one has freedom to choose $R$ to be any quantum memory system, and one can also replace $S_{\beta}(E|R)_{\omega}$ with the relative entropy of asymmetry, $\inf_{\sigma:\left[  H,\sigma \right] =0} D_{\alpha}(\rho\Vert\sigma)$. 

Equation~\eqref{eq:ourspeedlimit2} states that the orthogonalization time $\tau$ that appears in the Mandelstam-Tamm speed limit must go to infinity if the uncertainty $S_{\beta}(E|R)_{\omega}$ is less than one bit. This is a novel insight that does not follow from the Mandelstam-Tamm speed limit or other standard speed limits. Note that an information-theoretic approach to energy uncertainty (as opposed to the standard deviation that appears in the Mandelstam-Tamm speed limit) is necessary to obtain this insight, since the condition in \eqref{eq:ourspeedlimit2} is stated in terms of bits of information.

The above constraint can be generalized to multiple times, via our uncertainty relation. Define $\tau_K$ to be the time needed for the set $\{e^{-iH (k-1) \tau_K}\rho_A  e^{iH (k-1) \tau_k}\}_{k=1}^{K}$ to be composed of $K$ mutually orthogonal states. For this multi-time scenario, we obtain a generalization of \eqref{eq:ourspeedlimit2} as follows:
\begin{equation}
   \tau_K \to \infty \quad \text{if}\quad S_{\beta}(E|R)_{\omega} < (\log_2 K \text{ bits})\,.
  \label{eq:ourspeedlimit3}
\end{equation}
Hence, at the conceptual level our uncertainty relation not only constrains the $\tau$ appearing in quantum speed limits, but also constrains a more general scenario (i.e., a multi-time scenario) than is typically considered in quantum speed limits.

These insights inspire the following potential research directions: (1) Unify our information-theoretic constraint on $\tau$ in \eqref{eq:ourspeedlimit2} with standard quantum speed limits in order to obtain a stronger quantum speed limit, and (2) Formulate quantum speed limits for the more general scenario considered in \eqref{eq:ourspeedlimit3} involving multiple times.

\subsection{Randomness extraction}

Entropic uncertainty relations can certify that the bits extracted from a measurement are truly random from the perspective of an adversary. This adversary may have even supplied the quantum system to be measured and hence may have some background information about the state of this system. 

The applications of entropic uncertainty relations to randomness extraction are reviewed in \cite{entropic-ur-review}. One can specialize a R\'enyi entropic uncertainty relation to the min- and max-entropies, corresponding to setting $\alpha$ and $\beta$ to $\infty$ and $1/2$, in either order. For example, our uncertainty relation, in the discrete-time case, becomes
\begin{align}
   &S_{\max}(T|A)_{\kappa} + S_{\min}(E|R)_{\omega} \geq\log_2 |\mathcal{T}|, \label{eq:App-energy-time-uncertainty1}\\
   &S_{\min}(T|A)_{\kappa} + S_{\max}(E|R)_{\omega} \geq\log_2 |\mathcal{T}|, \label{eq:App-energy-time-uncertainty2}
\end{align}
where $S_{\min}$ and $S_{\max}$ refer to the min- and max-entropies, respectively.

The min-entropy $S_{\min}$ has operational significance in the task of randomness extraction via the Leftover Hashing Lemma~\cite{Impagliazzo1989, Bennett1995}. This lemma states that, if the initial min-entropy for the random variable $X$ is sufficiently large, then there exists a family of hash functions $\{ f_s \}_{s}$ such that the random variable $L = f_S(X)$, resulting from applying $f_S$ with a seed $S$ chosen uniformly at random, is approximately uniform and independent of $S$. 

Our uncertainty relations in \eqref{eq:App-energy-time-uncertainty1} and \eqref{eq:App-energy-time-uncertainty2} provide lower bounds on the min-entropy, which in turn allow one to certify randomness via the Leftover Hashing Lemma. One can either extract randomness from the energy variable $E$ or the time variable $T$, as suggested by \eqref{eq:App-energy-time-uncertainty1} and \eqref{eq:App-energy-time-uncertainty2}, respectively.

Randomness extraction from the polarization degree-of-freedom of single photons and pairs of photons was experimentally demonstrated in \cite{Vallone2014}. The randomness was certified using the min- and max-entropic uncertainty relation, and the experiment obtains approximately one bit of randomness per signal for single photons, and two bits of randomness per signal for pairs of photons. More recent work similarly employed the min- and max-entropic uncertainty relation but for the position and momentum observables, leading to a significantly higher rate of randomness \cite{Marangon2017}.  

Our work likewise allows one to achieve high rates of randomness (multiple bits per signal), as illustrated in the following example where randomness is extracted from the energy measurement. Suppose Alice receives coherent-state pulses from an untrusted source. This source randomly applies a time delay to each pulse, which adds a random phase to the coherent state. When Alice receives the pulse, she flips a coin. If she gets heads, she does a time measurement, which corresponds to extracting the phase of the coherent state by interfering it with a phase reference. If she gets tails, she does an energy measurement, which corresponds to measuring in the number basis with a photon-number-resolving detector. Alice uses her time measurement data to estimate $S_{\max} (T|A)$, which then allows her to lower bound $S_{\min}(E|R)$ via \eqref{eq:App-energy-time-uncertainty1} and hence to certify randomness extracted from the energy measurement via the Leftover Hashing Lemma. One can assume an adversary has possession of the $R$ system, and hence the randomness is certified to be secure even though the adversary has background information about the signal state. 

The above protocol has the advantage that coherent-state pulses are easily produced experimentally. In addition, multiple bits of randomness can be extracted per pulse. In particular, one can extract $\log_2 |\mathcal{T}|$ bits of randomness per pulse, where $|\mathcal{T}|$ can be chosen such that the coherent state produced by the source evolves to $|\mathcal{T}|$ fully distinguishable quantum states under the action of a time delay.  

\subsection{Quantum key distribution}

Protocols for quantum key distribution (QKD) involving the energy/time variables have previously been proposed, implemented, and analyzed \cite{Zhang2014, Niu2016, Qi2006}. These protocols considered the time variable in the context of photon arrival time, where the arrival time is measured with a time-resolving photon detector. Our uncertainty relation allows us to consider QKD protocols with other kinds of time encodings (besides arrival time). 

For example, in the context of coherent states, time delays map onto the phase of the coherent state, and for this example our uncertainty relation essentially becomes a number-phase uncertainty relation. We mentioned this above in the case of randomness extraction, and similarly one can formulate a number-phase QKD protocol. Here, Alice prepares a coherent state $\ket{\alpha}$ and then either encodes in time (for which she applies a random time delay) or in energy (for which she randomly prepares a number state with probability chosen according to the Poisson probability distribution associated with the coherent state $\ket{\alpha}$). She sends the resulting state over an insecure quantum channel to Bob, who then either tries to decode the time (by interfering the pulse with a phase reference) or the energy (by measuring with a photon-number-resolving detector). Potential benefits of this sort of QKD protocol would be the multiple bits of secure key obtained per pulse that Alice prepares.

Protocols involving other source states could be considered as well. For example, instead of preparing a coherent state, Alice could prepare a superposition of two number states: $ \ket{\psi} = \alpha \ket{n_1} + \beta \ket{n_2}$.  Similar to the above QKD protocol involving coherent states, Alice either applies a random time delay to $\ket{\psi}$ or she randomly prepares one of the number states $\ket{n_1}$ or $\ket{n_2}$ (with probabilities $|\alpha|^2$ and $|\beta |^2$, respectively). She sends the resulting state to Bob who either decodes the time or the energy.

In the aforementioned QKD protocols, Alice and Bob can either distill secret key out of their time data, energy data, or both. By bounding the information that the eavesdropper has about the secret key, our uncertainty relation can allow one to prove the security of such a protocol.

\section{%Appendix A:
Alternate version of the guessing game}

\label{appA}

Figure \ref{fig:guessing-game} describes a guessing game to better understand the trade-off between energy and time uncertainties. The game described earlier can be modified slightly with no change to the physical outcome. We first note that the result of an energy measurement on the state $\rho_A$ is the same as the result of an energy measurement on the state $e^{-i H t} \rho_A e^{i H t}$. This lets us restate the steps of the game as follows.

Alice applies one of $| \mathcal{T} |$ time evolutions on the state $\rho_A$ that she receives from Bob. She then flips a coin. If she obtains heads, she performs an energy measurement and sends the state back to Bob, who must then guess the outcome of Alice's energy measurement. If Alice obtains tails, she sends the state back to Bob, who must guess which of the $| \mathcal{T} |$  time evolutions was applied.

Everything stays the same as the game described in the main text, except for the fact that Alice applies a time evolution according to the Hamiltonian of system $A$ regardless of her coin toss outcome.

\section{%Appendix B:
Proof of Eq.~\eqref{eq:min-rel-ent-inequality}}\label{appB}

In this appendix, we prove the entropic uncertainty relation in \eqref{eq:min-rel-ent-inequality}, which we repeat here for convenience:
\begin{gather}
    S_{\alpha}(T|A)_{\kappa}+\inf_{\sigma:\left[  H,\sigma\right] =0}D_{\alpha}(\rho\Vert\sigma)\geq\log_2 | \mathcal{T} |
\end{gather}
Consider that%
\begin{equation}
\inf_{\sigma:[H,\sigma]=0}D_{\alpha}(\rho\Vert\sigma)=\inf
_{\sigma}D_{\alpha}\!\left(\rho\middle \Vert\sum_{\varepsilon}\Pi^{\varepsilon
}\sigma\Pi^{\varepsilon}\right).
\end{equation}
For a fixed state $\sigma$, consider that%
\begin{align}
 D_{\alpha}\!\left(\rho\middle\Vert\sum_{\varepsilon}\Pi^{\varepsilon}\sigma
\Pi^{\varepsilon}\right)  & =D_{\alpha}\!\left(\rho\otimes\pi_{T}\middle\Vert
\sum_{\varepsilon}\Pi^{\varepsilon}\sigma\Pi^{\varepsilon}\otimes\pi_{T}\right)\\
& =D_{\alpha}\!\left(\frac{1}{ | \mathcal{T} | }\sum_{t}e^{-iHt}\rho e^{iHt}%
\otimes|t\rangle\langle t|_{T}\middle\Vert\sum_{\varepsilon}\Pi^{\varepsilon}%
\sigma\Pi^{\varepsilon}\otimes\pi_{T}\right)\\
& =D_{\alpha}\!\left(\frac{1}{ | \mathcal{T} | }\sum_{t}e^{-iHt}\rho e^{iHt}%
\otimes|t\rangle\langle t|_{T}\middle\Vert\sum_{\varepsilon}\Pi^{\varepsilon}%
\sigma\Pi^{\varepsilon}\otimes I_{T}\right)  + \log_2 | \mathcal{T} | \\
& \geq\inf_{\sigma}D_{\alpha}\!\left(\frac{1}{ | \mathcal{T} | }\sum_{t}e^{-iHt}\rho
e^{iHt}\otimes|t\rangle\langle t|_{T}\middle\Vert\sigma\otimes I_{T}\right)+\log_2 | \mathcal{T} | \\
& =-S_{\alpha}(T|A)_{\kappa}+\log_2 | \mathcal{T} |.
\end{align}
The first equality follows because the relative entropy is invariant under tensoring in the maximally mixed state~$\pi_T$. The second inequality follows because relative entropy is invariant with respect to a controlled unitary, which here is $\sum_{t}e^{-iHt}\otimes|t\rangle\langle t|_{T}=e^{-iH\otimes\hat{T}}$, with $\hat{T} = \sum_{t} t \ket{t}\! \bra{t}$. Since the inequality holds for all states $\sigma$, and since $D_{\alpha} (\rho \Vert c \sigma) = D_{\alpha} (\rho \Vert \sigma) - \log_2 c$ for $c > 0$, we arrive at the claim in \eqref{eq:min-rel-ent-inequality}.

\section{Reducing Time Uncertainty with Quantum Memory}

\label{appTimeUncertainty}

In this appendix, we show that the tightening of the uncertainty relation in \eqref{eq:energy-time-uncertainty} with quantum memory can reduce entropic uncertainty in both energy and time. This follows by considering the memory system $R$ as a bipartite system $R_1 R_2$. We state \eqref{eq:energy-time-uncertainty} once again:
\begin{equation}
	S_{\alpha}(T|A)_{\kappa} + S_{\beta}(E|R)_{\omega} \geq\log_2 |\mathcal{T}|.
	\label{eq:appendix-memory-uncertainty}
\end{equation}

In the reformulated uncertainty relation, the Hamiltonian $H$ still acts non-trivially on system $A$ only, so that $H_{R_1 A} = I_{R_1} \otimes H_A$ and $H_A = \sum_{\varepsilon} \varepsilon \Pi_{A}^{\varepsilon}$. Thus the physical description of the composite system is unchanged. In the earlier description, the state $\kappa$ was defined on system $TA$ and $\omega$ on system $ER$. Now $\kappa$ is defined on system $TAR_1$ and $\omega$ on system $ER_2$. This leads to the entropic uncertainty relation \eqref{eq:appendix-memory-uncertainty} being rewritten as
\begin{equation}
	S_{\alpha} (T | AR_1)_{\kappa} + S_{\beta} ( E | R_2 )_{\omega}	\geq \log_2 | \mathcal{T} |.
\end{equation}

Thus we see now that our relation can be cast equivalently in the form above via a relabeling of the memory system $R$ as $R_1 R_2$. This implies that the reduction in uncertainty due to assistance of a quantum memory can manifest in both the entropic time and energy uncertainties.

\section{
Generalization to non-uniform time probabilities}

\label{appD}

In this appendix, we detail a particular generalization of the inequality in \eqref{eq:von-neumann-inequality}, which is \eqref{eq:energy-time-uncertainty} applied to the von Neumann entropies. The generalization involves a non-uniform distribution over the arbitrarily spaced times in the set $\mathcal{T}$. Instead of considering $| \mathcal{T} |$  uniformly weighted times in the state $\kappa_{TA}$, we can take the times to be weighted according to a probability mass function $p(k)$.

Consider a pure state $\ket{\psi}_{AR}$ with $\rho_A = \Tr_R(\dya{\psi})$. Also let
\begin{equation}
\kappa_{TA} =  \sum_{k=1}^{K} p(k)|k\rangle\langle k|_{T} \otimes e^{-iH t_k}\rho_A e^{iH t_k} ,\quad \text{with}\quad \kappa_{T} = \sum_{k=1}^{K}  p(k)|k\rangle\langle k| \, ,
\end{equation}
and
\begin{equation}
\omega_{ER}
\equiv\sum_{\varepsilon \in \mathcal{E}}
|\varepsilon\rangle\langle \varepsilon|_{E}\otimes\operatorname{Tr}%
_{A}\{\Pi^{\varepsilon}_A\psi_{AR}\}.
\end{equation}

Then
\begin{align}
S(E|R)_{\omega} &= D\Big(\rho_A \Vert \sum_{\varepsilon} \Pi^{\varepsilon}\rho_A \Pi^{\varepsilon}\Big)\\
& = D\Big(\kappa_{T} \otimes \rho_A   \Vert \sum_{\varepsilon} \kappa_{T} \otimes \Pi^{\varepsilon}\rho_A \Pi^{\varepsilon}  \Big)\\
& = D\Big(\kappa_{TA} \Vert \sum_{\varepsilon} \kappa_T \otimes \Pi^{\varepsilon}\rho_A \Pi^{\varepsilon} \Big)  \\
& = - S(\kappa_{TA}) - \Tr\Big(  \kappa_{TA} \log_2 \big(\sum_{\varepsilon} \kappa_T \otimes \Pi^{\varepsilon}\rho_A \Pi^{\varepsilon} \big) \Big) \\
& = - S(\kappa_{TA}) - \Tr\Big(  \kappa_{T} \log_2 \kappa_{T}\Big) - \Tr\Big(  \kappa_{A} \log_2 \big(\sum_{\varepsilon} \Pi^{\varepsilon}\rho_A \Pi^{\varepsilon}\big)\Big) \\
& = - S(\kappa_{TA}) + S(\kappa_{A}) - S(\kappa_{A}) - \Tr\Big(  \kappa_{A} \log_2 \sum_{\varepsilon} \Pi^{\varepsilon}\rho_A \Pi^{\varepsilon} \Big) + S(T)_{\kappa} \\
& = - S(T|A)_{\kappa} + D\Big(\kappa_{A} \Vert \sum_{\varepsilon} \Pi^{\varepsilon}\rho_A \Pi^{\varepsilon} \Big) + S(T)_{\kappa} ,
\end{align}
where the first equality can be shown, e.g., using Proposition \ref{prop1}: in the limit  $\alpha\rightarrow 1$, this  proposition implies $S(E|R)_\omega =\inf_{\sigma : [H_A,\sigma] = 0} D(\rho_A \Vert \sigma)$, which by the result of \cite{GMS09}, is equal to $\inf_{\sigma : [H_A,\sigma] = 0} D(\rho_A \Vert \sigma)= D(\rho_A \Vert  \sum_{\varepsilon} \Pi^{\varepsilon}\rho_A \Pi^{\varepsilon})$. 

Hence we obtain the following result:
\begin{align}
S(E|R)_{\omega} + S(T|A)_{\kappa} =  S(T)_{\kappa}  + D\Big(\kappa_{A} \Vert \sum_{\varepsilon} \Pi^{\varepsilon}\rho_A \Pi^{\varepsilon} \Big) .
\end{align}

Using the fact that relative entropy is non-negative,  this implies that 
\begin{align}
    S(E|R)_{\omega} + S(T|A)_{\kappa} \geq S(T)_{\kappa}  ,
\end{align}
where the inequality turns into an equality if and only if 
\begin{align}
\frac{1}{| \mathcal{T} |}\sum_{k = 1}^{K}e^{-iHt_k}\rho_A e^{iHt_k}  =  \sum_{\varepsilon} \Pi^{\varepsilon}\rho_A \Pi^{\varepsilon} \,.
\end{align}
Hence, for the von Neumann entropy case, our uncertainty relation is tight if and only if the above condition is satisfied. (One such case is when $[\rho_A, H]=0$.)

$S(T)_{\kappa} = S(\{ p(k)\})$ is the Shannon entropy of the probability distribution $\{p(k)\}$. So, in the special case of uniform probabilities, $p(k) = 1/| \mathcal{T} | $, we obtain again \eqref{eq:von-neumann-inequality}:
\begin{align}
S(E|R)_{\omega} + S(T|A)_{\kappa} \geq \log_2 | \mathcal{T} | \,.
\end{align}

\section{
Equivalence between \eqref{eq:energy-time-uncertainty} and \eqref{eq:min-rel-ent-inequality}}

\label{app:rel-ent-assym-to-cond-ent}

A consequence of the following proposition (by taking the $B$ system therein to be trivial) is that the quantity $S_{\beta}(E|R)_\omega $ in \eqref{eq:energy-time-uncertainty} is equal to the quantity $\inf_{\sigma : [H,\sigma] = 0} D_{\alpha}(\rho \Vert \sigma)$ in \eqref{eq:min-rel-ent-inequality}, whenever the state $\rho_{AR}$ is a pure state. As a result, the entropic uncertainty relations in \eqref{eq:energy-time-uncertainty} and \eqref{eq:min-rel-ent-inequality} are equivalent, whenever the state $\rho_{AR}$ is a pure state. The inequality in \eqref{eq:energy-time-uncertainty} holds for mixed $\rho_{AR}$ by purifying with an additional reference $R'$, invoking the result for pure bipartite states, and then applying the data processing inequality for conditional R\'enyi entropy \cite{Mueller2013} after a partial trace over $R'$. We note that this result is a generalization of a result in \cite{Coles2012}.

\begin{proposition} \label{prop1} Let $\psi_{ABC}$ be a pure tripartite state, and let $\{  \Pi_{A}^{j}\}
_{j}$ be a projective POVM (i.e., such that $\sum_{j}\Pi_{A}^{j}=I_{A}$ and $\Pi
_{A}^{j}\Pi_{A}^{j^{\prime}}=\delta_{j,j^{\prime}}\Pi_{A}^{j}$). Then%
\begin{equation}
\inf_{\sigma_{AB}}D_{\alpha}\!\left(\psi_{AB}\middle\Vert\sum_{j}\Pi_{A}%
^{j}\sigma_{AB}\Pi_{A}^{j}\right)=S_{\beta}(Z|C)_{\omega},
\end{equation}
where $\alpha\in[1/2,1)\cup(1,\infty]$, $\beta$ is such that $1/\alpha
+1/\beta=2$, and%
\begin{equation}
\omega_{ZC}\equiv\sum_{j}|j\rangle\langle j|_{Z}\otimes\operatorname{Tr}%
_{AB}\{\Pi_{A}^{j}\psi_{ABC}\},
\end{equation}
with $\{ \vert j \rangle_Z\}_j$ an orthonormal basis.
\end{proposition}

\begin{proof}
Our aim is to prove that%
\begin{equation}
\inf_{\sigma_{AB}}D_{\alpha}\!\left(\psi_{AB}\middle\Vert\sum_{j}\Pi_{A}%
^{j}\sigma_{AB}\Pi_{A}^{j}\right)
= -S_{\alpha}(Z|AB)_{\varphi
},\label{eq:rel-ent-to-cond-ent}%
\end{equation}
where%
\begin{align}
|\varphi\rangle_{AZBC}  & \equiv U_{A\rightarrow AZ}|\psi\rangle_{ABC},\\
U_{A\rightarrow AZ}  & \equiv\sum_{j}\Pi_{A}^{j}\otimes|j\rangle_{Z}.
\end{align}
Once this is established, it follows by duality of conditional sandwiched R\'enyi entropy \cite{Mueller2013,B13monotone} that%
\begin{equation}
-S_{\alpha}(Z|AB)_{\varphi}=S_{\beta}(Z|C)_{\varphi},
\end{equation}
and considering that%
\begin{align}
\operatorname{Tr}_{AB}\{|\varphi\rangle\langle\varphi|_{AZBC}\}  &
=\sum_{j,j^{\prime}}|j\rangle\langle j^{\prime}|_{Z}\otimes\operatorname{Tr}%
_{AB}\{\Pi_{A}^{j}\psi_{ABC}\Pi_{A}^{j^{\prime}}\}\\
& =\sum_{j,j^{\prime}}|j\rangle\langle j^{\prime}|_{Z}\otimes\operatorname{Tr}%
_{AB}\{\Pi_{A}^{j^{\prime}}\Pi_{A}^{j}\psi_{ABC}\}\\
& =\sum_{j,j^{\prime}}|j\rangle\langle j^{\prime}|_{Z}\otimes\delta
_{j,j^{\prime}}\operatorname{Tr}_{AB}\{\Pi_{A}^{j}\psi_{ABC}\}\\
& =\sum_{j}|j\rangle\langle j|_{Z}\otimes\operatorname{Tr}_{AB}\{\Pi_{A}%
^{j}\psi_{ABC}\}.
\end{align}
To this end, let $\sigma_{AB}$ be an arbitrary state. Then%
\begin{align}
D_{\alpha}\!\left(\psi_{AB}\middle\Vert\sum_{j}\Pi_{A}^{j}\sigma_{AB}\Pi_{A}%
^{j}\right)  & =
D_{\alpha}\!\left(U_{A\rightarrow AZ}\psi_{AB}U_{A\rightarrow
AZ}^{\dag}\middle \Vert U_{A\rightarrow AZ}\left[  \sum_{j}\Pi_{A}^{j}\sigma_{AB}%
\Pi_{A}^{j}\right]  U_{A\rightarrow AZ}^{\dag}\right)\\
& =D_{\alpha}\!\left(U_{A\rightarrow AZ}\psi_{AB}U_{A\rightarrow
AZ}^{\dag}\middle\Vert\sum_{j}\Pi_{A}^{j}\sigma_{AB}\Pi_{A}^{j}\otimes|j\rangle
\langle j|_{Z}\right)\\
& \geq D_{\alpha}\!\left(U_{A\rightarrow AZ}\psi_{AB}U_{A\rightarrow
AZ}^{\dag}\middle\Vert\sum_{j}\Pi_{A}^{j}\sigma_{AB}\Pi_{A}^{j}\otimes I_{Z}\right)\\
& \geq\inf_{\sigma_{AB}}D_{\alpha}(U_{A\rightarrow AZ}\psi
_{AB}U_{A\rightarrow AZ}^{\dag}\Vert\sigma_{AB}\otimes I_{Z})\\
& =-S_{\alpha}(Z|AB)_{\varphi}.
\end{align}
The first inequality follows because $\vert j \rangle \langle j \vert_Z \leq I_Z$ and from the property $D_{\alpha}(\rho \Vert \sigma) \geq 
D_{\alpha}(\rho \Vert \sigma')$ for $0 \leq \sigma \leq \sigma'$.
The second inequality follows because 
$\sum_{j}\Pi_{A}^{j}\sigma_{AB}\Pi_{A}^{j}$ is a state, and then we optimize over all possible states.
Since the above inequality holds for all states $\sigma_{AB}$, we find that
\begin{equation}
\inf_{\sigma_{AB}}D_{\alpha}\!\left(\psi_{AB}\middle \Vert\sum_{j}\Pi_{A}%
^{j}\sigma_{AB}\Pi_{A}^{j}\right)
\geq
-S_{\alpha}(Z|AB)_{\varphi}.
\end{equation}

Now let $\sigma_{AB}$ again be an arbitrary state. Then define the channel%
\begin{equation}
\overline{\Delta}_{ABZ}(\omega_{ABZ})=P_{ABZ}\omega_{ABZ}P_{ABZ}+\left(
I_{ABZ}-P_{ABZ}\right)  \omega_{ABZ}\left(  I_{ABZ}-P_{ABZ}\right)  ,
\end{equation}
where
$
P_{ABZ}\equiv\sum_{j}\Pi_{A}^{j}\otimes|j\rangle\langle j|_{Z}$.
Considering that
\begin{align}
P_{ABZ}U_{A\rightarrow AZ}  & =\left[  \sum_{j}\Pi_{A}^{j}\otimes
|j\rangle\langle j|_{Z}\right]  \left[  \sum_{j^{\prime}}\Pi_{A}^{j^{\prime}%
}\otimes|j^{\prime}\rangle_{Z}\right]  \\
& =\sum_{j,j^{\prime}}\Pi_{A}^{j}\Pi_{A}^{j^{\prime}}\otimes|j\rangle\langle
j|j^{\prime}\rangle_{Z}\\
& =\sum_{j}\Pi_{A}^{j}\otimes|j\rangle_{Z}\\
& =U_{A\rightarrow AZ},%
\end{align}
we find that%
\begin{align}
& D_{\alpha}(U_{A\rightarrow AZ}\psi_{AB}U_{A\rightarrow AZ}%
^{\dag}\Vert\sigma_{AB}\otimes I_{Z}) \notag \\
& \geq D_{\alpha}(\overline{\Delta}_{ABZ}(U_{A\rightarrow AZ}%
\psi_{AB}U_{A\rightarrow AZ}^{\dag})\Vert\overline{\Delta}_{ABZ}(\sigma
_{AB}\otimes I_{Z}))\\
& =D_{\alpha}(P_{ABZ}(U_{A\rightarrow AZ}\psi_{AB}U_{A\rightarrow
AZ}^{\dag})P_{ABZ}\Vert P_{ABZ}(\sigma_{AB}\otimes I_{Z})P_{ABZ})\\
& =D_{\alpha}\!\left(U_{A\rightarrow AZ}\psi_{AB}U_{A\rightarrow
AZ}^{\dag}\middle \Vert\sum_{j}\Pi_{A}^{j}\sigma_{AB}\Pi_{A}^{j}\otimes|j\rangle
\langle j|_{Z}\right)\\
& =D_{\alpha}\!\left(\psi_{AB}\middle \Vert\sum_{j}\Pi_{A}^{j}\sigma_{AB}\Pi
_{A}^{j}\right)\\
& \geq\inf_{\sigma_{AB}}D_{\alpha}\!\left(\psi_{AB}\middle \Vert\sum_{j}\Pi
_{A}^{j}\sigma_{AB}\Pi_{A}^{j}\right).
\end{align}
Since the above inequality holds for all states $\sigma_{AB}$, we find that%
\begin{equation}
-S_{\alpha}(Z|AB)_{\varphi}=\inf_{\sigma_{AB}} D_{\alpha}(U_{A\rightarrow AZ}\psi_{AB}U_{A\rightarrow AZ}^{\dag}\Vert
\sigma_{AB}\otimes I_{Z})\geq\inf_{\sigma_{AB}}D_{\alpha}%
(\psi_{AB}\Vert\sum_{j}\Pi_{A}^{j}\sigma_{AB}\Pi_{A}^{j}).
\end{equation}
Putting everything together implies \eqref{eq:rel-ent-to-cond-ent}.
\end{proof}

\section{
Extensions to continuous time and/or countable energy spectrum}

\label{app:contin-time-extension}

In this appendix, we provide a proof of the energy-time uncertainty relation
in three different cases:

\begin{enumerate}
\item Discrete time, Hamiltonian with countable spectrum and state $\rho_{A}$,
the latter two acting on a separable Hilbert space $\mathcal{H}_{A}$,

\item Continuous time, Hamiltonian with finite spectrum and state $\rho_{A}$,
the latter two acting on a finite-dimensional Hilbert space $\mathcal{H}_{A}$,

\item Continuous time, Hamiltonian with countable spectrum and state $\rho
_{A}$, the latter two acting on a separable Hilbert space $\mathcal{H}_{A}$.
\end{enumerate}

We begin with the first case. Let $H_{A}$ denote a Hamiltonian with a
countable spectrum, so that we can write it as%
\begin{equation}
H_{A}=\sum_{\varepsilon\in\mathcal{E}}\varepsilon\Pi_{A}^{\varepsilon},
\end{equation}
where the set $\mathcal{E}\subset\mathbb{R}$ is countable but bounded from
below, with smallest element $\varepsilon_{\min}$. Furthermore, suppose that
the support of $H_{A}$ is equal to the full Hilbert space $\mathcal{H}_{A}$.
We can then set an energy cutoff $E\geq\varepsilon_{\min}$, where $E$ is an
integer, and define the projection $\Pi_{A}^{E}=\sum_{\varepsilon
\in\mathcal{E}:\varepsilon\leq E}\Pi_{A}^{\varepsilon}$, so that $\Pi_{A}^{E}$
projects onto a finite-dimensional subspace of $\mathcal{H}_{A}$. Then define
the projected state $\rho_{A}^{E}$ from the original state $\rho_{A}$ as%
\begin{equation}
\rho_{A}^{E}=\Pi_{A}^{E}\rho_{A}\Pi_{A}^{E}+\operatorname{Tr}\{(I_{A}-\Pi
_{A}^{E})\rho_{A}\}\omega_{A}^{E},
\end{equation}
where $\omega_{A}^{E}$ is some arbitrary state supported on $\Pi_{A}^{E}$
(i.e., $\Pi_{A}^{E}\omega_{A}^{E}\Pi_{A}^{E}=\omega_{A}^{E}$). We also define
the following truncated Hamiltonian:%
\begin{equation}
H_{A}^{E}=\Pi_{A}^{E}H_{A}\Pi_{A}^{E}=\sum_{\varepsilon\in\mathcal{E}%
:\varepsilon\leq E}\varepsilon\Pi_{A}^{\varepsilon}.
\end{equation}
Note that the following limit holds%
\begin{equation}
\lim_{E\rightarrow\infty}\left\Vert \rho_{A}^{E}-\rho_{A}\right\Vert
_{1}=0.\label{eq:rho-state-limit}%
\end{equation}
Let the discrete times be given by $\mathcal{T}=\left\{  t_{1},\ldots
,t_{K}\right\}  $ where $K=\left\vert \mathcal{T}\right\vert $. Applying the
finite-dimensional energy-time entropic uncertainty relation in
\eqref{eq:min-rel-ent-inequality}, we find that%
\begin{equation}
\inf_{\sigma_{A}^{E}:\left[  H_{A}^{E},\sigma_{A}^{E}\right]  =0}D(\rho
_{A}^{E}\Vert\sigma_{A}^{E})+S(T|A)_{\kappa^{E}}\geq\log_{2}|\mathcal{T}%
|,\label{eq:finite-dim-unc-apply}%
\end{equation}
where $\sigma_{A}^{E}$ is an arbitrary state acting on the subspace onto which
$\Pi_{A}^{E}$ projects, and%
\begin{align}
\kappa_{TA}^{E} &  \equiv\frac{1}{|\mathcal{T}|}\sum_{k=1}^{K}|t_{k}%
\rangle\langle t_{k}|_{T}\otimes\rho_{A}^{E}(t),\\
\rho_{A}^{E}(t) &  \equiv e^{-iH_{A}^{E}t}\rho_{A}^{E}e^{iH_{A}^{E}t}.
\end{align}
Consider that the following limit holds%
\begin{equation}
\lim_{E\rightarrow\infty}\left\Vert \kappa_{TA}^{E}-\kappa_{TA}\right\Vert
_{1}=0,\label{eq:kappa-state-limit}%
\end{equation}
where%
\begin{equation}
\kappa_{TA}\equiv\frac{1}{|\mathcal{T}|}\sum_{k=1}^{K}|t_{k}\rangle\langle
t_{k}|_{T}\otimes e^{-iH_{A}t}\rho_{A}e^{iH_{A}t}.
\end{equation}
Consider that the inequality in \eqref{eq:finite-dim-unc-apply} is equivalent
to the following one:%
\begin{equation}
\inf_{\sigma_{A}^{E}}D(\rho_{A}^{E}\Vert\sum_{\varepsilon\in\mathcal{E}%
:\varepsilon\leq E}\Pi_{A}^{\varepsilon}\sigma_{A}^{E}\Pi_{A}^{\varepsilon
})+S(T|A)_{\kappa^{E}}\geq\log_{2}|\mathcal{T}|.
\end{equation}
So this means that, for an arbitrary positive-definite state $\sigma_{A}$
acting on the full separable Hilbert space $\mathcal{H}_{A}$, the following
inequality holds%
\begin{equation}
D(\rho_{A}^{E}\Vert\sum_{\varepsilon\in\mathcal{E}:\varepsilon\leq E}\Pi
_{A}^{\varepsilon}\sigma_{A}\Pi_{A}^{\varepsilon}/p_{E})+S(T|A)_{\kappa^{E}%
}\geq\log_{2}|\mathcal{T}|,
\end{equation}
where%
\begin{equation}
p_{E}=\sum_{\varepsilon\in\mathcal{E}:\varepsilon\leq E}\operatorname{Tr}%
\{\Pi_{A}^{\varepsilon}\sigma_{A}\}.
\end{equation}
Then we have that%
\begin{equation}
\lim_{E\rightarrow\infty}\left\Vert \sum_{\varepsilon\in\mathcal{E}%
:\varepsilon\leq E}\Pi_{A}^{\varepsilon}\sigma_{A}\Pi_{A}^{\varepsilon}%
/p_{E}-\sum_{\varepsilon\in\mathcal{E}}\Pi_{A}^{\varepsilon}\sigma_{A}\Pi
_{A}^{\varepsilon}\right\Vert _{1}=0.\label{eq:sigma-state-limit}%
\end{equation}
Now employing the limits in \eqref{eq:rho-state-limit},
\eqref{eq:kappa-state-limit}, and \eqref{eq:sigma-state-limit}, as well as the
limiting result for quantum relative entropy from \cite{Junge2018}, we
find that the following inequality holds for an arbitrary positive-definite
state $\sigma_{A}$:%
\begin{equation}
D(\rho_{A}\Vert\sum_{\varepsilon\in\mathcal{E}}\Pi_{A}^{\varepsilon}\sigma
_{A}\Pi_{A}^{\varepsilon})+S(T|A)_{\kappa}\geq\log_{2}|\mathcal{T}|.
\end{equation}
Since the inequality holds for an arbitrary positive-definite state
$\sigma_{A}$, and any positive semi-definite state can be approximated
arbitrarily well by a positive definite one, we can conclude that the
inequality above holds for an arbitrary state $\sigma_{A}$. Now, since we have
proven that the inequality holds for an arbitrary state $\sigma_{A}$, we can
conclude the following inequality:%
\begin{equation}
\inf_{\sigma_{A}}D(\rho_{A}\Vert\sum_{\varepsilon\in\mathcal{E}}\Pi
_{A}^{\varepsilon}\sigma_{A}\Pi_{A}^{\varepsilon})+S(T|A)_{\kappa}\geq\log
_{2}|\mathcal{T}|.
\end{equation}
This concludes the proof of the first case mentioned above.

\bigskip

We now turn to the second case mentioned above, in which the Hamiltonian has a
finite spectrum and the state $\rho_{A}$ acts on a finite-dimensional Hilbert
space, but there is a continuous time interval $\left[  0,T_{F}\right]  $. We
divide the time interval $\left[  0,T_{F}\right]  $ into $|\mathcal{T}%
|$\ equally sized bins, each of size $T_{F}/|\mathcal{T}|$, and we label each
bin by $t_{k}$ with $\mathcal{T}=\left\{  t_{1},\ldots,t_{K}\right\}  $ where
$K=\left\vert \mathcal{T}\right\vert $. We again start from the
finite-dimensional and (finite)\ discrete-time result from
\eqref{eq:min-rel-ent-inequality}, which implies that%
\begin{equation}
\inf_{\sigma_{A}:\left[  H_{A},\sigma_{A}\right]  =0}D(\rho_{A}\Vert\sigma
_{A})+S(T|A)_{\kappa}\geq\log_{2}|\mathcal{T}|,
\end{equation}
for%
\begin{align}
\kappa_{TA} &  \equiv\frac{1}{|\mathcal{T}|}\sum_{k=1}^{K}|t_{k}\rangle\langle
t_{k}|_{T}\otimes\rho_{A}(t),\\
\rho_{A}(t) &  \equiv e^{-iH_{A}t}\rho_{A}e^{iH_{A}t}.
\end{align}
Consider that%
\begin{align}
\inf_{\sigma_{A}:\left[  H_{A},\sigma_{A}\right]  =0}D(\rho_{A}\Vert\sigma
_{A}) &  =\inf_{\sigma_{A}}D(\rho_{A}\Vert\sum_{\varepsilon}\Pi_{A}%
^{\varepsilon}\sigma_{A}\Pi_{A}^{\varepsilon}),\\
S(T|A)_{\kappa} &  =-\sum_{k=1}^{K}D\Big(\frac{\rho_{A}(t_{k})}{|\mathcal{T}%
|}\Vert\overline{\rho}_{A}\Big),
\end{align}
where $\overline{\rho}_{A}=\frac{1}{\left\vert \mathcal{T}\right\vert }%
\sum_{k=1}^{K}\rho_{A}(t_{k})$. Then by introducing this scaling and adding
$-\log_{2}|\mathcal{T}|/T_{F}$ to the previous entropic trade-off, we find
that%
\begin{align}
\inf_{\sigma_{A}:\left[  H_{A},\sigma_{A}\right]  =0}D(\rho_{A}\Vert\sigma
_{A})+S(T|A)_{\kappa}-\log_{2}|\mathcal{T}|/T_{F} &  \geq\log_{2}%
|\mathcal{T}|-\log_{2}|\mathcal{T}|/T_{F}\\
&  =\log_{2}T_{F},
\end{align}
noting that this inequality holds for every such binning of the interval
$\left[  0,T_{F}\right]  $, with $|\mathcal{T}|$ arbitrarily large. After
applying \cite[Proposition~5]{furrer}, we find that%
\begin{equation}
\lim_{|\mathcal{T}|\rightarrow\infty}S(T|A)_{\kappa}-\log|\mathcal{T}%
|/T_{F}=-\int_{0}^{T_{F}}dt\ D(\rho_{A}(t)/T_{F}\Vert\overline{\rho}%
_{A})=s(T|A),
\end{equation}
where%
\begin{equation}
\overline{\rho}_{A}=\frac{1}{T_{F}}\int_{0}^{T_{F}}dt\ e^{-iH_{A}t}\rho
_{A}e^{iH_{A}t}.
\end{equation}
So taking the limit $|\mathcal{T}|\rightarrow\infty$ gives%
\begin{equation}
\inf_{\sigma_{A}:\left[  H_{A},\sigma_{A}\right]  =0}D(\rho_{A}\Vert\sigma
_{A})+s(T|A)\geq\log_{2}T_{F}.
\end{equation}
This concludes the proof of the inequality for the second case mentioned above.

\bigskip

The third case follows from a suitable combination of the first two. We can
first truncate the Hilbert space, take the continuous-time limit, and then
take the limit as the spectrum goes from finite to the full countable set.

\end{document}